\documentclass[3p,onecolumn]{elsarticle}
\usepackage{url}
\usepackage{amsmath}
\usepackage{amsfonts}
\usepackage{graphicx}
\usepackage{epstopdf}
\usepackage{enumitem}
\usepackage{tikz,graphics,color,epsf,float,epsf,caption,subcaption}
\usepackage{amssymb}
\usepackage{bm}
\usepackage{multirow}
\usepackage{natbib}
\usepackage{url}
\usepackage{amsthm}
\usepackage{color}
\usepackage{tikz}
\usepackage[ruled,vlined]{algorithm2e}
\usepackage{multicol}
\usepackage{multirow}
\usepackage{mathtools}
\usepackage[colorlinks = true, linkcolor = blue, urlcolor = blue, citecolor = blue, anchorcolor = blue]{hyperref}
\usepackage{physics}
\usepackage{lineno}
\newcommand{\ud}{\,\mathrm{d}}
\newcommand{\sgn}{\text{sgn}}

\journal{PRSA}

\begin{document}

\begin{frontmatter}

\title{The Generalized Fourier Transform: A Unified Framework for the Fourier, Laplace, Mellin and $Z$ Transforms}


\author[myprimaryaddress]{Pushpendra Singh\corref{mycorrespondingauthor}}
\cortext[mycorrespondingauthor]{Corresponding author. \\ Email address: spushp@nith.ac.in; pushpendrasingh@iitkalumni.org}

\author[mysecondaryaddress1]{Anubha Gupta}
\author[mysecondaryaddress2]{Shiv Dutt Joshi}

\address[myprimaryaddress]{Department of ECE, National Institute of Technology Hamirpur, Hamirpur, India.}
\address[mysecondaryaddress1]{SBILab, Department of ECE, Indraprastha Institute of Information Technology Delhi, Delhi, India}
\address[mysecondaryaddress2]{Department of EE, Indian Institute of Technology Delhi, Delhi, India}

\begin{abstract}
This paper introduces Generalized Fourier transform (GFT) that is an extension or the generalization of the Fourier transform (FT). The Unilateral Laplace transform (LT) is observed to be the special case of GFT. GFT, as proposed in this work, contributes significantly to the scholarly literature. There are many salient contribution of this work.  Firstly, GFT is applicable to a much larger class of signals, some of which cannot be analyzed with FT and LT. For example, we have shown the applicability of GFT on the polynomially decaying functions and super exponentials. Secondly, we demonstrate the efficacy of GFT in solving the initial value problems (IVPs). Thirdly, the generalization presented for FT is extended for other integral transforms with examples shown for wavelet transform and cosine transform. Likewise, generalized Gamma function is also presented. One interesting application of GFT is the computation of generalized moments, for the otherwise non-finite moments, of any random variable such as the Cauchy random variable. Fourthly, we introduce Fourier scale transform (FST) that utilizes GFT with the topological isomorphism of an exponential map. Lastly, we propose Generalized Discrete-Time Fourier transform (GDTFT). The DTFT and unilateral $z$-transform are shown to be the special cases of the proposed GDTFT. The properties of GFT and GDTFT have also been discussed.    

\end{abstract}

\end{frontmatter}

\section{Introduction}
The Fourier transform (FT) was proposed by the French mathematician Jean Baptiste Joseph Fourier (1768--1830) in a seminal manuscript submitted to the Institute of France in 1807 and his famous book about the analytic theory of heat \cite{FTH1,FTH2,FTH3}. Since the introduction of FT, it has become the most important mathematical tool for the modeling and analysis of an umpteen number of physical phenomena. As a consequence, it is utilized to solve problems of almost all fields of mathematics, engineering, science and technology \cite{SIAM1,FTA01,FQT01,FQT02,FQT03,FQT04,FQT05,FQT06,FQT07,FQT08,FQT09} to study and engineer signals and systems in various applications \cite{SH}. It is an essential tool for transmitting, processing and analyzing signals as well as is helpful in extraction and interpretation of the information. Although FT is widely applicable on a very large class of signals, it does have limitations. For example, FT of a signal $x(t)$ exists, if it is absolutely integrable, i.e., $\int_{-\infty}^{\infty} |x(t)| \, \mathrm{d}t<\infty$. There is a wide class of signals such as periodic sinusoidal signals, constant signals, and signum functions that do not satisfy this condition and hence, these signals are not Fourier transformable. This constraint is overcome by incorporating impulse functions within the framework of FT. This renders the condition of absolute integrability of FT as a sufficient, but not necessary condition. In other words, FT might exist even if this condition is not satisfied. FT is not only used to study signal properties, but also the system properties. It is also used in system modeling. In general, FT is very helpful in analyzing linear time-invariant (LTI) systems. Similar to signals, sometimes the impulse response of a system is not absolutely integrable. An implication of this is that the system is unstable. In order to study such systems, bilateral Laplace transform (BLT) was proposed \cite{LTH01,LTH02,LTH03,LTH04} as a generalization of FT to include a larger class of functions (or impulse response of systems) such as polynomial ($t^m \, \mathrm{u}(t), \, m>-1$) and  exponential ($e^{at} \, \mathrm{u}(t),\, a>0$) functions, where $\mathrm{u}(t)$ is the unit step function. Furthermore, unilateral LT, simply known as the LT, is a special case of BLT, where the transform is computed only for one half of the real axis. In fact, BLT of functions such as polynomial and exponential functions, periodic functions, constant functions, and signum functions defined over the entire real line, $-\infty<t<\infty$, do not exist. While impulse function helped in computing the FT of many of the above stated functions (such as periodic functions), FT and BLT do not exist for almost all polynomially decaying functions ($1/t^m,\, m\ge 1$). Hence, either the unilateral Laplace transform is used or the function is defined only for one half of the real line, such as $e^{at} \, \mathrm{u}(t),\, a>0$. In other words, one studies the causal LTI systems. 
The above discussion motivated us to look for a better generalization than BLT because it appears that BLT itself is limited to a great extent. This study makes a significant scholarly contribution and defines generalized FT (GFT) that is a true generalized framework of FT. GFT eliminates all the above stated problems and provides an effective mathematical tool to deal with these issues. Eventually, it turns out that both FT and BLT are special cases of the proposed GFT. \\

\noindent The salient contributions of this study are as follows:
\begin{enumerate}
    \item We propose generalized Fourier transform that is a true generalization of FT and show that the unilateral LT is a special case of the proposed GFT.
    
    \item The proposed GFT can analyze and synthesize a much larger set of signals compared to the Fourier and Laplace transforms.   
    
    \item Similar to Laplace transform, we demonstrate the efficacy of GFT in solving the initial value problems (IVPs). We also demonstrate the efficacy of GFT in solving those IVPs that cannot be solved using the LT.
    
    \item 
    We define the Fourier scale transform (FST) from the GFT using the topological isomorphism of exponential map.

    \item Fourier theory has been used to solve differential equations. However, since inception of the FT (1807), there is a perception in the literature that the IVPs cannot be solved using the FT. We propose a solution to this issue and demonstrate the efficacy of the FT in solving the IVPs. 
    
    \item We extend the GFT for the discrete-time systems and propose generalized discrete-time Fourier transform (GDTFT). The DTFT and unilateral $z$-transform are special cases of the proposed GDTFT. We also propose a solution of the difference equations with initial conditions using the DTFT, which is not available so far in the literature. 
    
    \item Finally, this new proposed GFT simplifies both the continuous-time and discrete-time signal analysis. For example: \textbf{(i)} in case of GFT, the region of convergence (ROC) is the entire $s$-plane, if there are no poles; otherwise, it is always right-side of the rightmost pole and FT exists if the ROC includes imaginary axis. 
    \textbf{(ii)} in case of GDTFT, the ROC is the entire $z$-plane, if there are no poles; otherwise, it is always outside of the outermost pole and DTFT exists if the ROC includes a unit circle. 
\end{enumerate}

Although, in this study we have provided theory related to one-dimensional signals, it can be easily extended for the multidimensional (e.g., image, video and time-space) signals. We have briefly mentioned the relevant equations regarding the same in Appendix-A. Rest of the study is organized as follows. In Section-2, we present the Generalized Fourier Transform, relation with FT, BLT, and LT, properties of GFT and solution of IVPs through GFT. In Section-3, we have extended this concept of GFT to the generalization of other Integral transforms, generalization of Gamma function, and discussed one application of GFT in the context of moment generating functions. The discrete-time counterpart, i.e., the generalized discrete-time Fourier transform (GDTFT) and its properties are discussed in Section-4. In the end, some concluding remarks along with the future directions are presented in Section-5.    

\section{The proposed Generalized Fourier Transform}\label{meth}
The pair of Fourier transform (FT) and inverse FT (IFT) of a signal $x(t)$, which satisfies the Dirichlet conditions, is defined as
\begin{equation}
\begin{aligned}
X(\omega)&=\int_{-\infty}^{\infty} x(t) \exp(-j\omega t) \ud t, \quad -\infty < \omega < \infty \\
\text{and \,\,} x(t)&={\frac{1}{2\pi}}\int_{-\infty}^{\infty}  X(\omega) \exp(j\omega t) \ud \omega, \quad -\infty < t < \infty,
\end{aligned}\label{FT1}
\end{equation}
subject to the existence of the above integrals. This time-frequency pair is also denoted by $x(t)\rightleftharpoons X(f)$, where $\omega=2 \pi f$, $\omega$ denotes the angular frequency in radians/sec, and $f$ denotes the frequency in Hz.
The bilateral Laplace transform (BLT), a generalization of FT, and its inverse are defined as
\begin{equation}
\begin{aligned}
	X(s)&=\int_{-\infty}^{\infty} x(t) \exp(-s\, t) \ud t,\\
	\text{ and } \quad x(t) &=\frac{1}{2\pi j}\int_{\sigma-j\infty}^{\sigma + j\infty} X(s) \exp(s\, t) \, \mathrm{d} s,
\end{aligned}\label{LT1}
\end{equation}
respectively, where the symbol `\textit{s}' denotes the complex frequency $s=\sigma+j\omega$. BLT is the FT of the signal multiplied with exponentially decaying function. In other words, it is the FT of $x(t)e^{-\sigma t}$.
The unilateral LT is the special case of BLT, when the considered signals (or impulse response of linear time-invariant systems) are causal, i.e., time support is only for $0\le t< \infty$ and hence, $X(s)=\int_{0}^{\infty} x(t) \exp(-s t) \ud t$.  
The signum (`\textit{sgn}') and unit step functions are defined \cite{SH} as
\begin{equation}
	sgn(t) =
	\begin{cases}
		1,  \quad t>0, \\
		0, \quad t = 0,\\
		-1   \quad t<0.
	\end{cases} \quad \text{ and, } \quad u(t) =
\begin{cases}
1,  \quad t\ge 0, \\
0,   \quad t<0,
\end{cases} \quad \text{ or } \quad u(t) =
\begin{cases}
1,  \quad t> 0, \\
\frac{1}{2}, \quad t=0,\\
0,   \quad t<0.
\end{cases} 
\label{FTexa1}
\end{equation} 

If two functions have identical values except at a countable number of points, their integration is also identical because Lebesgue measure of a set of countable points is always zero. These functions are said to belong to the same equivalent class of functions. In other words, from the theory of integral calculus, both the above definitions of the unit step function in \eqref{FTexa1} are equivalent.     

Motivated with the theory of Fourier transform \eqref{FT1} and the Laplace transform \eqref{LT1}, we define a generalized version of the FT, namely the  generalized Fourier transform (GFT) $X(\omega,\sigma,\tau)$ of a signal $x(t)$, as 
\begin{equation}
\begin{aligned}
X(\omega,\sigma,\tau)&=\int_{-\infty}^{\infty} x(t) \, h(\sigma,\tau,t) \exp(-j\omega t) \ud t,\\
x(t) \, h(\sigma,\tau,t)&={\frac{1}{2\pi}}\int_{-\infty}^{\infty}  X(\omega,\sigma,\tau) \exp(j\omega t) \ud \omega.
\end{aligned}\label{GFT1}
\end{equation}

\noindent A number of interesting special cases arise out of GFT, given as below:
\begin{enumerate}
	\item Fourier transform (FT): if $h(\sigma,\tau,t)=1$.
	\item Fourier Cosine transform: if $h(\sigma,\tau,t)=1$ and $x(t)$ is an even function.
    \item Fourier Sine transform: if $h(\sigma,\tau,t)=1$ and $x(t)$ is an odd function.
   	\item Short-time Fourier transform (STFT): if $h(\sigma,\tau,t)=w(t-\tau)$, where $w(t)$ is a window function.
   	
   	\item Unilateral LT (in short LT): if $h(\sigma,\tau,t)=\exp(-\sigma t)\mathrm{u}(t)$ and $\sigma>c$, where $c$ is a constant that depends on the signal $x(t)$.
   	\item Bilateral LT: if $h(\sigma,\tau,t)=\exp(-\sigma t)$, where $\mathrm{u}(t)$ is the unit step function as defined in \eqref{FTexa1} and $\sigma>c$, where $c$ is a constant that depends on the signal $x(t)$.

	\item Generalized Fourier transform (GFT): if $h(\sigma,\tau,t, p)= (\sgn(t) \, t)^p \exp(-\sigma \, \sgn(t) \, t)$, where $\sgn(t)$ is the sign function \eqref{FTexa1}. We designate this proposed representation as GFT. The development of this transform is the fundamental contribution of this study.
\end{enumerate}

\subsection{The GFT and its relation with Fourier and Laplace transforms}
Our first aim of this study is to define and explore the GFT of a signal $x(t)$, hereby, presented as 
\begin{equation}
	\begin{aligned}
		\mathcal{G}\{x(t)\}=X(\omega,\sigma)&=\int_{-\infty}^{\infty} x(t) \exp(-\sigma \, \sgn(t) \, t) \exp(-j\omega t) \ud t,
	\end{aligned}\label{GFT2}
\end{equation}
where the value of $\sigma$ is such that the integral \eqref{GFT2} converges to a finite value, i.e., $\sigma\in \text{ROC}$ or the region of convergence that depends on the signal $x(t)$. We also denote \eqref{GFT2} as $X(\omega,\sigma)=X(s,s^*)$ where complex frequency $s=\sigma+j\omega$, its complex conjugate $s^*=\sigma-j\omega$, and $\sgn(t) \, t=|t|$. 
Using the `sgn' function as defined in \eqref{FTexa1}, we write \eqref{GFT2} as
\begin{equation}
	\begin{aligned}
		X(\omega,\sigma)&=\int_{-\infty}^{0} x(t) \exp(\sigma t) \exp(-j\omega t) \ud t+\int_{0}^{\infty} x(t) \exp(-\sigma t) \exp(-j\omega t) \ud t.
	\end{aligned}\label{GFT3}
\end{equation}
We further simplify \eqref{GFT3} and write as
\begin{equation}
	\begin{aligned}
		X(s,s^*) =\mathfrak{X}(s^*)+X(s)&=\int_{0}^{\infty} x(-t) \exp(-s^* t) \ud t+\int_{0}^{\infty} x(t) \exp(-s t) \ud t,
	\end{aligned}\label{GFT4}
\end{equation}
where $X(s)$ is the Laplace transform defined as 
\begin{equation}
	\begin{aligned}
		X(s)&=\int_{-\infty}^{\infty} x(t) \, \mathrm{u}(t) \exp(-s t) \ud t=\int_{0}^{\infty} x(t) \exp(-s t) \ud t,
	\end{aligned}\label{GFT5}
\end{equation}
and $\mathfrak{X}(s^*)$ is the new term, designated as the complementary LT (CLT), and hereby defined as 
\begin{equation}
	\begin{aligned}
	\mathfrak{X}(s^*)&=\int_{-\infty}^{\infty} x(t) \, \mathrm{u}(-t) \exp(s^* t) \ud t=\int_{-\infty}^{0} x(t) \exp(s^* t) \ud t=\int_{0}^{\infty} x(-t) \exp(-s^* t) \ud t.
	\end{aligned}\label{GFT6}
\end{equation}
If there is a delta function at the origin, it can be included either in \eqref{GFT5} or \eqref{GFT6}, or with half amplitude in both. 

From \eqref{GFT2}, \eqref{GFT5} and \eqref{GFT6}, the following observations are in order:
\begin{enumerate}
\item If $x(t)$ is an even function, $X(s^*)=\mathfrak{X}(s^*)$.
\item If $x(t)$ is an odd function, $X(s^*)=-\mathfrak{X}(s^*)$.
\item If function $x(t)$ is a sum of the even and odd functions, i.e., $x(t)=x_e(t)+x_o(t)$ and $x(-t)=x_e(t)-x_o(t)$, then $X(s)=X_e(s)+X_o(s)$ and $\mathfrak{X}(s^*)=\mathfrak{X}_e(s^*)+\mathfrak{X}_o(s^*)$.
\item If $\sigma=0$, then GFT \eqref{GFT2} becomes the FT \eqref{FT1}, provided FT exists for the given signal.
\item If $x(t)=0$ for $t<0$, then GFT \eqref{GFT2} becomes the LT \eqref{GFT5}.
\item If we simply consider $|\sgn(t)|$ in \eqref{GFT2}, then it becomes the bilateral LT.    
\end{enumerate}

Using the IFT \eqref{FT1} and GFT \eqref{GFT3}, we obtain 
\begin{equation}
	\begin{aligned}
		x(t) \exp(-\sigma \, \sgn(t) \, t)	&=\frac{1}{2\pi}\int_{-\infty}^{\infty} X(\omega,\sigma) \exp(j\omega t) \, \mathrm{d} \omega.
	\end{aligned}\label{IGFT00}
\end{equation}

One simple way to recover/synthesize the original signal $x(t)$ is by evaluating \eqref{IGFT00} at $\sigma=0$. Other ways to synthesize $x(t)$ as derived from \eqref{IGFT00} are
\begin{align}
x(t)=\frac{1}{2\pi}\int_{-\infty}^{\infty} X(\omega,\sigma) \exp(\sigma \, \sgn(t) \, t) \exp(j\omega t) \, \mathrm{d} \omega, \quad \forall \, t \in \mathbb{R},   
\end{align}
or
\begin{equation}
	x(t) =
	\begin{dcases}
		\frac{1}{2\pi}\int_{-\infty}^{\infty} X(s) \exp(\sigma t) \exp(j\omega t) \, \mathrm{d} \omega, & t \in [0,\,\infty) \\
		\frac{1}{2\pi}\int_{-\infty}^{\infty} \mathfrak{X}(s^*) \exp(-\sigma t) \exp(j\omega t) \, \mathrm{d} \omega, & t \in (-\infty,\,0].
	\end{dcases} 
\label{IGFT2}
\end{equation}
On substituting $s=\sigma + j\omega$ in \eqref{IGFT2}, we obtain 
\begin{equation}
	x(t) =
	\begin{dcases}
		\frac{1}{2\pi j}\int_{\sigma-j\infty}^{\sigma + j\infty} X(s) \exp(s t) \, \mathrm{d} s,  & t \in [0,\,\infty), \\
		\frac{1}{2\pi j}\int_{\sigma-j\infty}^{\sigma + j\infty} \mathfrak{X}(s^*) \exp(-s^* t) \, \mathrm{d} s, & t \in (-\infty,\,0],
	\end{dcases} 
\label{IGFT3}
\end{equation}
where the value of $x(t)$ at $t=0$ can be obtained by either one of the above formulations. 

\vspace{2 mm}\noindent\textbf{Remark 1.} There are many classes of signals (e.g. one-sided exponential signals $\exp(at)\,\mathrm{u}(t)$ and $\exp(-at)\,\mathrm{u}(-t)$ with $a>0$) for which FT does not exist, but BLT exists. There are also classes of signals (e.g. constant, periodic and two sided signals such as $1,\, \cos(\omega_c t), \, \sin(\omega_c t)$ and $\exp(-a|t|)$ with $a>0$) for which BLT does not exist, but FT exists. This suggests that the LT is not a true generalization of the FT. On the other hand, the proposed GFT can analyze and synthesize all these classes of signals and therefore, presents a true generalization of FT. 

\vspace{2 mm} For examples, first we consider the signal $x(t)=1$, its FT is given by  $X(\omega)=2\pi \delta(\omega)$, its LT is $X(s)=\frac{1}{s}, \, \text{for \, Re}\{s\}=\sigma>0$, but its BLT does not exist. The proposed GFT is given by
\begin{equation}
	\begin{aligned}
		X(s,s^*)=\frac{1}{s^*}+\frac{1}{s}&=\frac{2\sigma}{\sigma^2+\omega^2}, \\
			\text{ and } \lim_{\sigma \to 0} \frac{2\sigma}{\sigma^2+\omega^2}&=2\pi \delta(\omega),					
	\end{aligned}\label{ex1}
\end{equation} 
because $\int_{-\infty}^{\infty}\frac{2\sigma}{\sigma^2+\omega^2} \, \mathrm{d}\omega=2\tan^{-1}\left(\frac{\omega}{\sigma}\right)\Big|_{-\infty}^{\infty}=2 \pi$. This clearly shows that the proposed GFT is a true generalization of the FT. Consider the GFT pair, $\sgn(t) e^{-\sigma \, \sgn(t)\, t} \rightleftharpoons \frac{-2j \omega}{\sigma^2+\omega^2} $. On using the duality property of FT ($x(t)\rightleftharpoons X(f)$ and $X(t)\rightleftharpoons x(-f)$ or $X(t) \rightleftharpoons 2\pi x(-\omega)$), we obtain $\frac{1}{\pi}\frac{t}{\sigma^2+t^2} \rightleftharpoons -j \, \sgn(\omega) e^{-\sigma \, \sgn(\omega)\, \omega}$, which is otherwise difficult to derive directly. For $\lim \sigma \to 0$, it is the FT of the Hilbert transform (HT) kernel. Similarly, $\frac{1}{\pi}\frac{\sigma}{\sigma^2+t^2} \rightleftharpoons e^{-\sigma \, \sgn(\omega)\, \omega}$ and for $\lim \sigma \to 0$, it is the FT of the Dirac delta function. Therefore, the function $\frac{1}{\pi}\frac{t}{\sigma^2+t^2}$ is the HT of the function $\frac{1}{\pi}\frac{\sigma}{\sigma^2+t^2}$.    Next, we present Table-\ref{tab1} where some signals and their GFT are presented. 

{\renewcommand{\arraystretch}{1.5}%
\begin{table}[h]	
	\begin{center}
	\caption{Some signals and their GFT with information whether FT and/or BLT exists. To obtain FT from the given GFT, first separate out the real and imaginary parts by substituting $s=\sigma+j\omega$ and then consider $\lim \sigma \to 0$. We observe that the ROC is the entire $s$-plane, if there are no poles; otherwise, it is always right-side of the rightmost pole. Furthermore, the FT exists, if the ROC includes imaginary axis, $s= j\omega$.} \label{tab1}
	\begin{tabular}{|l|l|l|}
		\hline
		Signals & Proposed GFT and ROC & FT/BLT \\ \hline
	$\delta(t)$	&  $1$,      all $s$ & Yes/Yes  \\
	$\delta(t-t_0), \, t_0>0$	&  $\exp(-s t_0)=\exp[-(\sigma+j\omega) t_0]$,      all $s$ & Yes/Yes  \\	
	$\delta(t+t_0), \, t_0>0$	&  $\exp(-s^* t_0)=\exp[-(\sigma-j\omega) t_0]$,      all $s$ & Yes/Yes  \\	
	1	&  $\frac{1}{s^*}+\frac{1}{s}=\frac{2\sigma}{\sigma^2+\omega^2}$,      $\text{Re}\{s\}>0$ & Yes/No  \\ 
	$\sgn(t)$	&  $\frac{-1}{s^*}+\frac{1}{s}=\frac{-j2\omega}{\sigma^2+\omega^2}$,      $\text{Re}\{s\}>0$ & Yes/No  \\ 
	$\mathrm{u}(t)$	&  $\frac{1}{s}=\frac{\sigma}{\sigma^2+\omega^2}-\frac{j\omega}{\sigma^2+\omega^2}$,      $\text{Re}\{s\}>0$ & Yes/Yes  \\ 
	$\mathrm{u}(-t)$	&  $\frac{1}{s^*}=\frac{\sigma}{\sigma^2+\omega^2}+\frac{j\omega}{\sigma^2+\omega^2}$,      $\text{Re}\{s\}>0$ & Yes/Yes  \\ 		
	$e^{-at}$	&  $\frac{1}{s^*-a}+\frac{1}{s+a}=\frac{2\sigma}{\sigma^2-(a+j\omega)^2}$,   $\text{Re}\{s\}>\text{max}\left(\text{Re}\{a\} ,\text{Re}\{-a\}\right)$ & No/No \\ 
	
	$e^{-a|t|}$	&  $\frac{1}{s^*+a}+\frac{1}{s+a}=\frac{2\sigma+2a}{\sigma^2+\omega^2+2a\sigma+a^2}$,   $\text{Re}\{s\}>\text{Re}\{-a\}$ & No/No \\

	$e^{j\omega_0 t}$	&  $\frac{1}{s^*+j\omega_0}+\frac{1}{s-j\omega_0}=\frac{2\sigma}{\sigma^2+(\omega-\omega_0)^2}$,   $\text{Re}\{s\}>0$ & Yes/No \\	
	
	$\cos(\omega_0 t)$	&  $\frac{s^*}{(s^*)^2+(\omega_0)^2}+\frac{s}{s^2+\omega^2_0}=\frac{\sigma}{\sigma^2+(\omega+\omega_0)^2}+\frac{\sigma}{\sigma^2+(\omega-\omega_0)^2}$,   $\text{Re}\{s\}>0$ & Yes/No  \\	
	$\sin(\omega_0 t)$	&  $\frac{-\omega_0}{(s^*)^2+(\omega_0)^2}+\frac{\omega_0}{s^2+\omega^2_0}=-j\left[\frac{-\sigma}{\sigma^2+(\omega+\omega_0)^2}+\frac{\sigma}{\sigma^2+(\omega-\omega_0)^2}\right]$,   $\text{Re}\{s\}>0$ & Yes/No \\
	$t^m,\, m>-1$	&  $\Gamma(m+1)\left[\frac{(-1)^m}{(s^*)^{m+1}}+\frac{1}{s^{m+1}}\right]$,       $\text{Re}\{s\}>0$ & No/No \\ 
	$|t|^m,\, m>-1$	&  $\Gamma(m+1)\left[\frac{1}{(s^*)^{m+1}}+\frac{1}{s^{m+1}}\right]$,       $\text{Re}\{s\}>0$ & No/No\\
	\hline
	\end{tabular} 
\end{center}
\end{table}}

\subsection{Properties of the GFT}
\noindent\textbf{Remark 2.} It is pertinent to note that all the properties of the FT are valid for the proposed GFT.
For example, if we consider two signals and their GFTs using  FT as:  $x_1(t)\exp(-\sigma \, \sgn(t) \, t) \rightleftharpoons X_1(\omega,\sigma)$ and $x_2(t)\exp(-\sigma \, \sgn(t)\, t) \rightleftharpoons X_2(\omega,\sigma)$, then,
     \begin{align}
     	x_1(t)\exp(-\sigma \, \sgn(t) \, t) \times x_2(t)\exp(-\sigma \, \sgn(t) \, t)  \rightleftharpoons \frac{1}{2\pi} \left[X_1(\omega,\sigma) \star X_2(\omega,\sigma)\right], \nonumber\\
     	\left[x_1(t)\times x_2(t)\right]\exp(-\sigma \, \sgn(t) \, t)  \rightleftharpoons \frac{1}{2\pi} \left[X_1(\omega,\sigma_1) \star X_2(\omega,\sigma_2)\right],\label{conv1}
     \end{align}
     where $\star$ is the convolution operation and, $\sigma_1$ and $\sigma_2$ are suitably selected such that $\sigma_1+\sigma_2=\sigma$, yielding one degree of freedom in choosing the attenuation parameter depending on the application in hand. Similarly, we can show that
     \begin{align}
      x_1(t)\exp(-\sigma \, \sgn(t) \, t) \star x_2(t)\exp(-\sigma \, \sgn(t) \, t) \rightleftharpoons X_1(\omega,\sigma) \times X_2(\omega,\sigma),\nonumber\\
      x_1(t) \star x_2(t)\exp(-\sigma \, \sgn(t) \, t) \rightleftharpoons X_1(\omega,0) \times X_2(\omega,\sigma),\label{conv2}\\
      x_1(t) \exp(-\sigma \, \sgn(t) \, t) \star x_2(t) \rightleftharpoons X_1(\omega,\sigma) \times X_2(\omega,0). \nonumber
     \end{align}
     	
     Using these properties of \eqref{conv1} and \eqref{conv2}, we can study those linear time-invariant (LTI) systems whose impulse response is not absolutely integrable. That is, we can analyze unstable systems by considering the range of damping factor $\sigma$ and capture the behavior of systems by taking $\lim \sigma \to 0.$          
\newtheorem{thm}{Theorem}
\begin{thm}
Let $\tilde{x}(t)$ be a periodic function of period $T$, i.e., $\tilde{x}(t+T)=\tilde{x}(t),\, \forall\, t$. If GFT of $\tilde{x}(t)$ exists, then
\begin{align}
\tilde{X}(s,s^*) =\left(\frac{e^{-s^* T}}{1-e^{-s^* T}} \right)\int_0^T \tilde{x}(t) e^{s^* t} \, \mathrm{d}t + \left(\frac{1}{1-e^{-s T}}\right)\int_0^T \tilde{x}(t) e^{-s t} \, \mathrm{d}t, \quad \text{for}\quad \text{Re}\{s\}>0.   \label{thrm0}   
\end{align}
\end{thm}
\begin{proof}
Let $x(t)$ be a time-limited signal defined as
\begin{equation}
	x(t) =
	\begin{cases}
		\tilde{x}(t),  \quad 0\le t \le T, \\
		0, \quad \text{otherwise,}
	\end{cases}
\label{Thrm1}
\end{equation}
where $\tilde{x}(t)$ is a periodic signal that can also be written as   
\begin{align}
\tilde{x}(t)=\sum_{m=-\infty}^\infty x(t-mT)=\sum_{m=-\infty}^{-1} x(t-mT)+\sum_{m=0}^\infty x(t-mT). \label{Thrm2}
\end{align}
The GFT of \eqref{Thrm2} yields
\begin{align}
 \tilde{X}(s,s^*) &= \left(\dots+\int_{-3T}^{-2T} x(t+3T) e^{s^*t} \, \mathrm{d}t+\int_{-2T}^{-T} x(t+2T) e^{s^*t} \, \mathrm{d}t+ \int_{-T}^0 x(t+T) e^{s^*t} \, \mathrm{d}t \right) \nonumber \\ 
 & + \left(\int_0^T x(t) e^{-st} \, \mathrm{d}t+ \int_T^{2T} x(t-T) e^{-st} \, \mathrm{d}t+\int_{2T}^{3T} x(t-2T) e^{-st} \, \mathrm{d}t+ \dots\right),
\label{Thrm3}
\end{align}
which can be simplified as
\begin{align}
\tilde{X}(s,s^*) =\left( \dots +e^{-3s^*T}+ e^{-2s^*T} + e^{-s^*T}\right)\int_{0}^T x(t) e^{s^*t} \, \mathrm{d}t+ \left( 1 + e^{-sT} + e^{-2sT}+\dots \right)\int_{0}^T x(t) e^{-st} \, \mathrm{d}t. \label{Thrm4}
\end{align}
Thus, we obtain \eqref{thrm0},
which completes the proof. 
\end{proof}

Next, we derive the properties of GFT. Here, ROC is the entire $s$ plane, if there are no poles; otherwise, ROC is always right-side of the rightmost pole.
\begin{enumerate}
    \item Linearity:
    \begin{align}
        \text{ if } \quad \mathcal{G}\{x(t)\} & =X(s,s^*)=\mathfrak{X}(s^*)+X(s), \quad \text{ with ROC } R\\
       \text{ then } \quad \mathcal{G}\left\{\sum_{i} a_i x_i(t)\right\} & =\sum_i a_iX_i(s,s^*)= \sum_i a_i\left[\mathfrak{X}_i(s^*)+X_i(s)\right], \quad  a_i \in \mathbb{C},
    \end{align}
    where ROC is at least $\bigcap^n_{i=0} R_i$ and $R_i$ is the ROC corresponding to $X(s,s^*)$.
    \item The GFT of differentiation in time-domain:
    \begin{align}
    \mathcal{G}\{x'(t)\} & =\left[-s^*\mathfrak{X}(s^*)+x(0)\right]+\left[sX(s)-x(0)\right],\\
    \mathcal{G}\{x''(t)\} & =\left[(s^*)^2 \mathfrak{X}(s^*)-s^*\,x(0)+x'(0)\right]+\left[s^2X(s)-s\,x(0)-x'(0)\right],
\end{align}
and for the $m$th order differentiation, we obtain
\begin{align}
    \mathcal{G}\{x^{(m)}(t)\} & =\left[(-s^*)^m \mathfrak{X}(s^*)+\sum_{i=1}^m (-s^*)^{(m-i)} x^{(i-1)}(0)\right]+\left[s^m X(s)-\sum_{i=1}^m s^{(m-i)} x^{(i-1)}(0)\right], \label{dp1}
\end{align}
where $x'(0)=\frac{\mathrm{d}}{\mathrm{d}t} x(t)\Big|_0$ and $x^{(m)}(0)=\frac{\mathrm{d}^m}{\mathrm{d}t^m} x(t)\Big|_0$, and ROC is at least $R$.

\item The GFT of the running integral of a signal $x(t)$:  
\begin{align}
 \mathcal{G}\left\{\int_0^t x(\tau) \, \mathrm{d}\tau \right\} & =\frac{-\mathfrak{X}\left(s^*\right)}{s^*}+\frac{X\left(s\right)}{s},\\
 \mathcal{G}\left\{\int_t^0 x(\tau) \, \mathrm{u}(-t) \, \mathrm{d}\tau +\int_0^t x(\tau) \, \mathrm{u}(t) \, \mathrm{d}\tau \right\} & =\frac{\mathfrak{X}\left(s^*\right)}{s^*}+\frac{X\left(s\right)}{s},
\end{align}
where ROC is at least $R\,\bigcap\,\{\text{Re}\{s\}>0\}$.

\item The GFT of $x(t)$ multiplied by various functions or scaled in amplitude:
\begin{align}
 \mathcal{G}\{e^{-at}\,x(t)\} & =\mathfrak{X}(s^*-a)+X(s+a); \quad \text{ROC is } R-|\text{Re}\{a\}|\\
 \mathcal{G}\{e^{-a|t|}\,x(t)\} & =\mathfrak{X}(s^*+a)+X(s+a); \quad \text{ROC is } R+\text{Re}\{a\}\\
  \mathcal{G}\{t^{m}\,x(t)\} & =\frac{\mathrm{d}^m}{\mathrm{d}{s^*}^m }\mathfrak{X}(s^*)+(-1)^m\frac{\mathrm{d}^m}{\mathrm{d}s^m }X(s);\quad \text{ROC is } R\\  
   \mathcal{G}\{|t|^{m}\,x(t)\} & =(-1)^m\frac{\mathrm{d}^m}{\mathrm{d}{s^*}^m }\mathfrak{X}(s^*)+(-1)^m\frac{\mathrm{d}^m}{\mathrm{d}s^m }X(s) \quad \text{ROC is } R.   
\end{align}
It is also observed that 
\begin{align}
 \frac{\mathrm{d}^m}{\mathrm{d}{\sigma}^m } X(s,s^*) & =\frac{\partial^m}{\partial{s^*}^m } X(s,s^*)+(-1)^m\frac{\partial^m}{\partial{s}^m } X(s,s^*)=\frac{\mathrm{d}^m}{\mathrm{d}{s^*}^m }\mathfrak{X}(s^*)+(-1)^m\frac{\mathrm{d}^m}{\mathrm{d}s^m }X(s),\\
 \frac{\mathrm{d}^m}{\mathrm{d}{\omega}^m } X(s,s^*) & ={j^m}\frac{\partial^m}{\partial{s^*}^m } X(s,s^*)+{(-j)^m}\frac{\partial^m}{\partial{s}^m } X(s,s^*)={j^m}\frac{\mathrm{d}^m}{\mathrm{d}{s^*}^m }\mathfrak{X}(s^*)+{(-j)^m}\frac{\mathrm{d}^m}{\mathrm{d}s^m }X(s).
\end{align}

\item The GFT of $x(t)$ divided by $t$, provided $\lim_{t \to 0}\left(\frac{x(t)}{t}\right)$ exists:
\begin{align}
    \mathcal{G}\left\{\frac{x(t)}{t}\right\} & =-\int_{s^*}^\infty \mathfrak{X}(v) \, \mathrm{d}v+\int_{s}^\infty {X}(u) \, \mathrm{d}u,\\
    \mathcal{G}\left\{\frac{x(t)}{|t|}\right\} & =\int_{s^*}^\infty \mathfrak{X}(v) \, \mathrm{d}v+\int_{s}^\infty {X}(u) \, \mathrm{d}u.
\end{align}

\item Time scaling:
\begin{align}
 \mathcal{G}\{x(at)\} & =\frac{1}{a}\mathfrak{X}\left(\frac{s^*}{a}\right)+\frac{1}{a}X\left(\frac{s}{a}\right), \quad a>0, \quad \text{ROC is } R/a\\
 \mathcal{G}\{x(-at)\} & =\frac{1}{a}X\left(\frac{s^*}{a}\right)+\frac{1}{a}\mathfrak{X}\left(\frac{s}{a}\right), \quad a>0 \quad \text{ROC is } R/a.
\end{align}

\item Time reversal:
\begin{align}
    \mathcal{G}\{x(-t)\} & =X(s^*)+\mathfrak{X}(s), \quad \text{ROC is } R.
\end{align}

\item Time delay: Consider writing $x(t)$ as:  $x(t)=x(t)\,\mathrm{u}(-t)+x(t)\,\mathrm{u}(t)$. This implies that $x(t-\sgn(t) \, t_0)=x(t+t_0)\,\mathrm{u}(-t-t_0)+x(t-t_0)\,\mathrm{u}(t-t_0)$ and thus, 
\begin{align}
    \mathcal{G}\{x(t-\sgn(t) \, t_0)\} & =e^{-s^* \, t_0}\mathfrak{X}(s^*)+e^{-s \, t_0}X(s), \quad t_0\ge0, \quad \text{ROC is } R,
\end{align}
and
\begin{align}
    \mathcal{G}\{x(t-t_0)\} & =\int_{-\infty}^{0} x(t-t_0) \, e^{s^*t} \, \mathrm{d}t+\int_{0}^{\infty}x(t-t_0)\,e^{-st} \, \mathrm{d}t, \quad t_0\in \mathbb{R},\nonumber\\
    & =e^{s^* \, t_0}\int_{-\infty}^{-t_0} x(t) \, e^{s^*t} \, \mathrm{d}t+e^{-s \, t_0}\int_{-t_0}^{\infty}x(t)\,e^{-st} \, \mathrm{d}t.
\end{align}

\end{enumerate}

\subsection{Solution of IVPs using the proposed GFT}
In this subsection, we demonstrate the utility of GFT for solving the initial value problems (IVP's). Let us consider the IVP as
\begin{align}
    x''(t)+\omega_c^2 \, x(t) = c_1 \, \delta(t-t_0) + c_2 \, \delta(t+t_0), \quad x(0)=a_1, \, x'(0)=a_2 \, \omega_c, \, t_0\ge 0, \label{de1}
\end{align}
where constants $a_1,\, a_2,\, c_1,\, c_2 \in \mathbb{C}$, $t,\,\omega_c \in \mathbb{R}, \, \omega_c\ne 0$. Using the differentiation property of the proposed GFT \eqref{dp1}, we obtain
\begin{align}
s^2 X(s) & - s\,x(0) -x'(0)+\omega_c^2\, X(s)=c_1\, e^{-t_0\, s},  \quad \text{ for } t\ge 0\\
 X(s)&= a_1 \, \frac{s}{s^2+\omega_c^2} + a_2 \, \frac{\omega_c}{s^2+\omega_c^2} + c_1 \, e^{-s t_0}
 \frac{1}{s^2+\omega_c^2}, \quad \text{ for } t\ge 0,\\
 x(t)&= \left[a_1 \, \cos(\omega_c t) + a_2 \, \sin(\omega_c t)\right] \mathrm{u}(t) + \frac{c_1}{\omega_c} \, \sin\left((t-t_0)\omega_c \right) \, \mathrm{u}(t-t_0),
 \label{de2}
\end{align}
and 
\begin{align}
(s^*)^2 \mathfrak{X}(s^*) & - s^*\,x(0) +x'(0)+\omega_c^2\, \mathfrak{X}(s^*)=c_2\, e^{-t_0\, s^*}, \text{ for } t\le 0,\\
 \mathfrak{X}(s^*)&= a_1 \, \frac{s^*}{(s^*)^2+\omega_c^2} - a_2 \, \frac{\omega_c}{(s^*)^2+\omega_c^2} + c_2 \, e^{- t_0 \, s^*}
 \frac{1}{(s^*)^2+\omega_c^2}, \quad \text{ for } t\le 0,\\
 x(t)&= \left[a_1 \, \cos(\omega_c t) + a_2 \, \sin(\omega_c t)\right] \mathrm{u}(-t) - \frac{c_2}{\omega_c} \, \sin\left((t+t_0)\omega_c \right) \mathrm{u}(-t-t_0).
 \label{de3}
\end{align}
Therefore, from \eqref{de2} and \eqref{de3}, we can write, $x(t)$, $\forall t\in \mathbb{R}$, as
\begin{align}
x(t)&= a_1 \, \cos(\omega_c t) + a_2 \, \sin(\omega_c t) +\frac{c_1}{\omega_c} \, \sin\left((t-t_0)\omega_c \right) \mathrm{u}(t-t_0) - \frac{c_2}{\omega_c} \, \sin\left((t+t_0)\omega_c \right) \mathrm{u}(-t-t_0).
 \label{de4}
\end{align}
It is to be noted that this kind of IVP's cannot be solved using the traditional Fourier and Laplace transforms. 

\subsubsection{Solution of the IVPs using the FT}
We, hereby, observe that when $\sigma = 0$ (or $ \lim \sigma \to 0$), then GFT will become the FT with $s=j\omega$, $s^*=-j\omega$ and thus, FT can be used for finding the solution of IVPs by adopting the differentiation property discussed in \eqref{dp1}, provided the FT of the functions under consideration exist. Thus, to obtain the solution of IVPs, we derive and propose the following derivative property of the FT as
    \begin{align}
    \mathcal{G}\{x'(t)\} & =\left[s\mathfrak{X}(s)+x(0)\right]+\left[sX(s)-x(0)\right], \label{dpft00}\\
    \mathcal{G}\{x''(t)\} & =\left[s^2 \mathfrak{X}(s)+s\,x(0)+x'(0)\right]+\left[s^2X(s)-s\,x(0)-x'(0)\right], \label{dpft0}
\end{align}
and for the $m$th derivatives, we obtain
\begin{align}
    \mathcal{G}\{x^{(m)}(t)\} & =\left[s^m \mathfrak{X}(s)+\sum_{i=1}^m s^{(m-i)} x^{(i-1)}(0)\right]+\left[s^m X(s)-\sum_{i=1}^m s^{(m-i)} x^{(i-1)}(0)\right], \label{dpft1}
\end{align}
where, $s=j\omega$, $\mathfrak{X}(s)=\mathfrak{X}(\omega)=\int_{-\infty}^0 x(t) \, e^{-j\omega t} \, \mathrm{d}t$, $X(s)=X(\omega)=\int_{0}^\infty x(t) \, e^{-j\omega t} \, \mathrm{d}t$, $x'(0)=\frac{\mathrm{d}}{\mathrm{d}t} x(t)\Big|_0$ and $x^{(m)}(0)=\frac{\mathrm{d}^m}{\mathrm{d}t^m} x(t)\Big|_0$. Here, the main observation from \eqref{dpft00}, \eqref{dpft0} and \eqref{dpft1} is that, in general, the initial conditions for positive value of time cancel the initial conditions for negative value of time. This is the main difficulty in solving the IVPs in the traditional way using the FT. Thus, to solve the IVPs, we should avoid the cancellation and retain the initial conditions as presented in \eqref{dpft00}, \eqref{dpft0} and \eqref{dpft1}.  

\subsection{Generalized Fourier transform and polynomially decaying functions}
The Fourier and Laplace transforms of the polynomially decaying functions do not exist. Therefore, in order to tackle the polynomial decay in the function, our next goal is to define and explore the generalized Fourier transform (GFT) as 
\begin{equation}
	\begin{aligned}
		X(\omega,\sigma,p)&=\int_{-\infty}^{\infty} x(t) \, |t|^p \exp(-\sigma \, |t|) \exp(-j\omega t) \ud t,
	\end{aligned}\label{pGFT2}
\end{equation} 
which can be written as
\begin{equation}
	\begin{aligned}
		X(s,s^*,p)=\mathfrak{X}(s^*,p)+X(s,p)&=\int_{0}^{\infty} x(-t) \, t^p \exp(-s^* t) \ud t+\int_{0}^{\infty} x(t)  \, t^p \exp(-s t) \ud t.
	\end{aligned}\label{pGFT4}
\end{equation}
Using the proposed definition \eqref{pGFT2}, we can easily include the polynomially decaying class of functions. 
For example, considering $x(t)=\frac{1}{t^m}, \, m \ge 1$, we obtain
\begin{equation}
	\begin{aligned}
		X(s,s^*,p)= \Gamma(p-m+1)\left[\frac{1}{(-1)^m}\frac{1}{(s^*)^{p-m+1}}+\frac{1}{s^{p-m+1}}\right], \quad p>m-1.
	\end{aligned}\label{pGFT5}
\end{equation}
To derive the Fourier and Laplace transforms of the HT kernel, we set $m=1$ in \eqref{pGFT5} and obtain  
\begin{equation}
	\begin{aligned}
		X(\omega,\sigma,p)= \Gamma(p)\left[\frac{1}{s^{p}}-\frac{1}{(s^*)^{p}}\right]=\Gamma(p)\left[\frac{(\sigma-j\omega)^{p}-(\sigma+j\omega)^{p}}{(\sigma^2+\omega^2)^p}\right].
	\end{aligned}\label{pGFT6}
\end{equation}
We can further simplify \eqref{pGFT6} by considering $\lim \sigma \to 0$ and obtain  
\begin{equation}
	\begin{aligned}
		X(\omega,p)= \Gamma(p) \left[\frac{(-j\,\omega)^{p}-(j\,\omega)^{p}}{\omega^{2p}}\right],
	\end{aligned}\label{pGFT7}
\end{equation}
which can be written as
\begin{equation}
	\begin{aligned}
		X(\omega,p)=\Gamma(p) \frac{(-2j)}{ \omega^p} \sin\left(\frac{p\pi}{2}\right) , \, \omega >0 \text { and } X(\omega,p)=\Gamma(p) \frac{2j}{\omega^p} \sin\left(\frac{p\pi}{2}\right), \, \omega <0.
	\end{aligned}\label{pGFT8}
\end{equation}
In \eqref{pGFT7}, for $p=1$, $X(\omega)=-2j/\omega$, which is the FT of the signum function. From \eqref{pGFT8} we observe that, $\lim_{p \to 0} X(\omega,p)=- j \, \pi \, \sgn(\omega)$. Thus, we obtain the FT of the HT kernel as $\frac{1}{\pi t} \rightleftharpoons - j \, \sgn(\omega)$. This is the direct derivation of the FT of the HT kernel without using the duality property.      

\subsection{The GFT of super-exponential functions}
The Fourier and Laplace transforms of more than exponentially growing functions (such as $e^{at^m},\, m>1,\, a>0$) do not exist. 
Therefore, to include this class of functions, we define GFT as 
\begin{equation}
	\begin{aligned}
		X(\omega,\sigma,p,q)&=\int_{-\infty}^{\infty} x(t) \, |t|^p \exp(-\sigma \, |t|^q) \exp(-j\omega t) \ud t, \qquad p \ge 0, \, q \ge 0.
	\end{aligned}\label{peGFT2}
\end{equation} 
Using the proposed definition \eqref{peGFT2}, we can easily include the class of functions growing at a rate faster than the exponential functions. However, this may be limited in applications because evaluation of such integrals would be quite tedious. Although GFT \eqref{peGFT2} is quite general, it is not sufficient to deal with the class of functions that explodes in a limited time support such as $\tan(\omega t)$, $\cot(\omega t)$, $\sec(\omega t)$, $\csc(\omega t)$, and $e^{1/(1-t)}$.

We can also obtain the GFT of a function $x(t)$ by using the Taylor series expansion (TSE) and computing the GFT of each term of the series, provided the resultant series $X(s,s^*)$ in the GFT domain converges. For example, we consider the TSE of $x(t)=e^{t}=\sum_{m=0}^\infty \frac{1}{m!} t^{m}$ and obtain its GFT as $\mathcal{G}\{e^{t}\}=X(s,s^*)=\sum_{m=0}^\infty \left[\frac{(-1)^m}{(s^*)^{m+1}}+\frac{1}{s^{m+1}}\right]=\frac{1}{s^*+1}+\frac{1}{s-1}, \quad \text{Re}\{s\}>0$ and $|s|>1$, or $\text{Re}\{s\}>1$.

\section{Extension of the proposed methodology}
\subsection{Integral transforms}
The integral transform (IT) of a function $x(t)$ is defined as
\begin{align}
    X(\omega,\sigma)=\int_{t_1}^{t_2} x(t) \, K(\omega,\sigma,t) \ud t, \label{IT1}
\end{align}
where $K(\omega,\sigma,t)$ is the kernel of the considered IT. The original signal $x(t)$ can be recovered, if there exists an inverse kernel $K^{-1}(\omega,\sigma,t)$ as
\begin{align}
  x(t)  =\int_{u_1}^{u_2}  X(\omega,\sigma) \, K^{-1}(\omega,\sigma,t) \ud \omega. \label{IT2}
\end{align}
The IT \eqref{IT1} and its inverse \eqref{IT2} may not exist for all possible signals. Therefore, similar to the GFT \eqref{GFT1}, we can define generalized IT (GIT) as
\begin{align}
    X(\omega,\sigma)=\int_{t_1}^{t_2} x(t)\, h(\sigma,t) \, K(\omega,\sigma,t) \ud t, \\
    x(t)\, h(\sigma,t)  =\int_{u_1}^{u_2}  X(\omega,\sigma) \, K^{-1}(\omega,\sigma,t) \ud \omega, \label{IT3}
\end{align}
where we may chose, e.g., $h(\sigma,t)=\exp(-\sigma|t|)$ in such a fashion that IT exists for the larger class of functions. There are many integral transforms such as fractional FT (FrFT), continuous wavelet transform (CWT), linear canonical transform (LCT), Abel transform, Fourier sine and cosine transform, Hankel transform, Hartley transform, Hermite transform, Hilbert transform, Jacobi transform, Laguerre transform, Legendre transform, Mellin transform, Poisson kernel, Radon Transform, Weierstrass transform, S-transform, etc. 
Similar to the GFT \eqref{GFT2} and \eqref{pGFT2}, we can extend the proposed methodology and include that class of signals which are not in $L^p(\mathbb{R}), \, p\in(0,\infty)$. For example, the CWT of a signal $x(t)$ is defined if $x(t) \in L^2(\mathbb{R})$ and the original can be recovered by inverse CWT \cite{awt0,awt1,awt2, awt3,awt4}. We hereby define the CWT of a signal $x(t)$ such that $x(t)\not\in L^2(\mathbb{R})$ and $x(t) \exp(-\sigma \, |t|) \in L^2(\mathbb{R})$ for some $\sigma>\sigma_0$ as
\begin{equation}
	\begin{aligned}
		X_\psi(a,b,\sigma)&=\int_{-\infty}^{\infty} x(t)  \exp(-\sigma \, |t|) \, \psi^*_{a,b}(t) \, \mathrm{d} t, \qquad a \ne 0,
	\end{aligned}\label{CWT1}
\end{equation}
and $X_\psi(a,b,\sigma)=0$ for $a=0$, where $\psi_{a,b}(t)=\frac{1}{\sqrt{|a|}}\psi\left(\frac{t-b}{a}\right)$ is a family of daughter wavelets and $\psi(t)\in L^2( \mathbb{R})$ is a mother wavelet function with zero-mean and finite energy. The original signal $x(t)$ can be recovered as
\begin{equation}
	\begin{aligned}
	x(t)	&=\frac{\exp(\sigma \, |t|)}{C_\psi}\int_{-\infty}^{\infty}\int_{-\infty}^{\infty}  X_\psi(a,b,\sigma) \, \psi_{a,b}(t) \, \mathrm{d} b \frac{1}{a^2} \mathrm{d} a, \qquad C_\psi \ne 0,
	\end{aligned}\label{CWT2}
\end{equation}
where the wavelet $\psi(t)\rightleftharpoons \Psi(\omega)$ satisfies the admissibility condition \cite{awt0,awt1,awt2, awt3,awt4}
\begin{align}
 C_\psi= \int_{-\infty}^{\infty} \frac{|\Psi(\omega)|^2}{|\omega|}  \mathrm{d} \omega < \infty.
\end{align}
We present another example of an integral transform where the Fourier cosine transform (FCT) and inverse FCT pairs can be defined for the larger classes of signals as
\begin{align}
 X_c(\omega,\sigma)&= \sqrt{\frac{2}{\pi}}\int_{0}^{\infty} x(t)\, e^{-\sigma t} \cos(\omega t)  \mathrm{d} t, \quad \omega\ge0, \, \sigma>0,\\
 x(t)&= \sqrt{\frac{2}{\pi}}\int_{0}^{\infty} X_c(\omega,\sigma)\, e^{\sigma t} \cos(\omega t)  \mathrm{d} \omega, \quad t\ge0,
\end{align}
and when $\lim \sigma \to 0$, this becomes the original FCT and IFCT. 
\subsection{Fourier Scale Transform (FST)}
The Mellin transform (MT) and inverse MT (IMT) are derived from the BLT \eqref{LT1} by the change of the variable as $e^{-t}=\tau\implies t= -\ln(\tau)$ and is defined as
\begin{equation}
\begin{aligned}
	Y(s)&=\int_{0}^{\infty} y(\tau) \, \tau^{s-1} \ud \tau,\\
	\text{ and } \quad y(t) &=\frac{1}{2\pi j}\int_{\sigma-j\infty}^{\sigma + j\infty} Y(s) \, \tau^{-s} \, \mathrm{d} s,
\end{aligned}\label{MT1}
\end{equation}
where $x(-\ln(\tau))=y(\tau)$. We can define the lower and upper partial MT as
\begin{align}
    Y(s)=Y_L(s,\tau_c)+Y_U(s,\tau_c)=\int_{0}^{\tau_c} y(\tau) \, \tau^{s-1} \ud \tau+\int_{\tau_c}^{\infty} y(\tau) \, \tau^{s-1} \ud \tau. \label{IMT}
\end{align}
The MT and IMT presented in \eqref{MT1} are essentially Fourier/Laplace counterparts processed through an isomorphism. MT is a different realization of the BLT by the topological isomorphic exponential mapping, $\exp \colon (\mathbb{R,+}) \to (\mathbb{R^{+},\cdot})$. Therefore all the limitations of the BLT, as already discussed, are enforced to the MT as well. For example, MT of a simple function $y(\tau)=1$ does not exists.

To overcome these limitations, using the GFT \eqref{GFT2}, by the change of the variable as $e^{-t}=\tau\implies t= -\ln(\tau)$, we define the FST and inverse FST (IFST) as
\begin{equation}
\begin{aligned}
	Y(s,s^*)=Y_L(s,1)+Y_U(-s^*,1)&=\int_{0}^{1} y(\tau) \, \tau^{s-1} \ud \tau + \int_{1}^{\infty} y(\tau) \, \tau^{-s^*-1} \ud \tau,\\
	y(t) &=\frac{1}{2\pi j}\int_{\sigma-j\infty}^{\sigma + j\infty} Y_L(s,1) \, \tau^{-s} \, \mathrm{d} s, \quad 0\le t \le 1,\\
	\text{ and } \quad y(t) &=\frac{1}{2\pi j}\int_{\sigma-j\infty}^{\sigma + j\infty} Y_U(-s^*,1) \, \tau^{s^*} \, \mathrm{d} s, \quad 1\le t < \infty.
\end{aligned}\label{FST1}
\end{equation}

Here, we consider a very interesting example of the MT. The gamma function is the MT of $y(\tau)=e^{-\tau}$ and is defined as
\begin{align}
    \Gamma(s)=\int_{0}^{\infty} e^{-\tau} \, \tau^{s-1} \ud \tau=(s-1)\Gamma(s-1),
    \label{GF01}
\end{align}
which can be written as the sum of the lower and upper partial gamma functions as
\begin{align}
    \Gamma(s)=\Gamma_L(s,\tau_c)+\Gamma_U(s,\tau_c)=\int_{0}^{\tau_c} e^{-\tau} \, \tau^{s-1} \ud \tau+\int_{\tau_c}^{\infty} e^{-\tau} \, \tau^{s-1} \ud \tau. \label{GF02}
\end{align}
The gamma function \eqref{GF01} can also be realized by the BLT of the function $x(t)=e^{-e^{-t}}$ as
\begin{align}
\Gamma(s)=\int_{-\infty}^{\infty} e^{-e^{-t}} \, e^{-s\,t} \, \mathrm{d} t= \frac{\Gamma(s+1)}{s}, \quad s\ne 0. \label{GF03}  
\end{align}

\noindent Next, we consider the FST of the function $y(\tau)=e^{-\tau}$ and obtain
\begin{align}
    Y(s,s^*)=\int_{0}^{1} e^{-\tau} \, \tau^{s-1} \ud \tau+\int_{1}^{\infty} e^{-\tau} \, \tau^{-s^*-1} \ud \tau,
    \label{MFT01}
\end{align}
which can be written in terms of the incomplete gamma functions as
\begin{align}
    Y(s,s^*)=\Gamma_L(s,1)+\Gamma_U(-s^*,1).
    \label{MFT02}
\end{align}

\noindent Next, we derive the properties of the proposed FST that are listed below.
\begin{enumerate}
\item Scaling:
\begin{align}
    \text{ if } \quad \mathcal{F}\{y(t)\} & =Y_L(s,1)+Y_U(-s^*,1),\nonumber\\
    \text{ then } \quad \mathcal{F}\{y(at)\} & =\frac{1}{a^s}Y_L(s,1)+a^{s^*}Y_U(-s^*,1).
\end{align}

\item Multiplication by $t^m$:
\begin{align}
    \mathcal{F}\{t^m y(t)\} & =Y_L(m+s,1)+Y_U(m-s^*,1),
\end{align}
which implies, $s \mapsto s+m$, and, $s^* \mapsto s^*-m$. 

\item The FST of derivatives:
    \begin{align}
    \mathcal{F}\{y'(t)\} & =\left[y(1)-(s-1)Y_L(s-1,1)\right]+\left[(s^*+1)Y_U(-s^*-1,1)-y(1)\right],\\
    \mathcal{G}\{y''(t)\} & =\left[y'(1)-(s-1)y(1)+(s-1)(s-2)Y_L(s-2,1)\right]\nonumber \\ 
    &+\left[-y'(1)-(s^*+1)y(1)+(s^*+1)(s^*+2)Y_U(-s^*-2,1)\right],\\
    \mathcal{F}\{t y'(t)\} & =\left[y(1)-sY_L(s,1)\right]+\left[s^*Y_U(-s^*,1)-y(1)\right],\\
    \mathcal{F}\{t^2 y''(t)\} & =\left[y'(1)-(s+1)y(1)+s(s+1)Y_L(s,1)\right] \nonumber\\
    &+\left[-y'(1)-(s^*-1)y(1)+s^*(s^*-1)Y_U(-s^*,1)\right].
    \end{align}

It is also observed that 

\begin{align}
\frac{\partial^m}{\partial{s}^m } Y(s,s^*)&=\int_0^1 (\ln{t})^m \, y(t)\, t^{s-1} \ud t,\\ \text{ and } \quad
    \frac{\partial^m}{\partial{s^*}^m } Y(s,s^*)&= \int_1^{\infty} (-\ln{t})^m \, y(t)\, t^{-s^*-1} \ud t.
\end{align}

\end{enumerate}

{\renewcommand{\arraystretch}{1.5}%
\begin{table}[h]	
	\begin{center}
	\caption{Some signals and their FST along with the information on whether MT exists. We observe that the ROC is the entire $s$-plane, if there are no poles; otherwise, it is always right-side of the rightmost pole. Also, the FT exists, if the ROC includes the imaginary axis $s= j\omega$.} \label{tabMT1}
	\begin{tabular}{|l|l|l|}
		\hline
		Signals & Proposed FST and ROC & MT \\ \hline
	$\delta(t-t_0), \, t_0\in(0,1]$	&  $t_0^{s-1}$,      all $s$ & Yes  \\	
	
	$\delta(t-t_0), \, t_0\in[1,\infty)$	&  $t_0^{-s^*-1}$,      all $s$ & Yes  \\	
	
	$e^{-at}\,\mathrm{u}(t)$, $a>0$	&  $\frac{1}{a^s}\Gamma_L(s,a)+{a^{s^*}}\Gamma_U(-s^*,a)$ & Yes\\ 
	
	${e^{-t}}\,\mathrm{u}(t)$	&  $\Gamma_L(s,1)+ \Gamma_U(-s^*,1)$, $\text{Re}\{s\}>0$ & Yes\\

	$\frac{1}{e^t-1}\,\mathrm{u}(t)$	&  $\sum_{\ell=1}^{\infty} \left[ \ell^{-s}\,\Gamma_L(s,\ell)+\ell^{s^*}\, \Gamma_U(-s^*,\ell)\right]$, $\text{Re}\{s\}>0$ & Yes\\

	$\mathrm{u}(t)$	&  $\frac{1}{s}+\frac{1}{s^*}$,      $\text{Re}\{s\}>0$ & No  \\

	$t^{a}\,\mathrm{u}(t)$	&  $\frac{1}{s+a}+\frac{1}{s^*-a}$, $\text{Re}\{s\}>\text{Re}\{a\}$ & No\\

	\hline
	\end{tabular} 
\end{center}
\end{table}}

\subsection{Generalized Gamma function}
The gamma function was introduced by Leonhard Euler (1707--1783) to generalize the factorial to non integer values. The gamma function (GF) is defined (Euler, 1730) as
\begin{align}
    \Gamma(s)=\int_0^\infty t^{s-1} \, e^{-t} \, \mathrm{d} t, \quad \text{Re}\{s\}>0, \label{gamma1}
\end{align}
where the above integral converges for $t \in [0,\infty)$.
Using the iteration of the gamma identity as $\Gamma(s+1)=s\Gamma(s)$, this function can be defined on the entire real axis and can be extended to the entire complex plane except for $s=0$ and for $s$ as negative integers, i.e., $(s=0,-1,-2, \dots)$ \cite{InGamma}. 
Motivated from above and the theory of GFT, we define the generalized gamma function (GGF) as
\begin{align}
    G(s)=\int_{-\infty}^\infty t^{s-1} \, e^{-\sgn(t) \, t} \, \mathrm{d} t, \quad \text{Re}\{s\}>0, \label{gamma2}
\end{align}
where the above integral converges for all $t \in \mathbb{R}$. The GGF \eqref{gamma2} can be written as
\begin{align}
    G(s)=\Gamma_c(s)+\Gamma(s)=\int_{0}^\infty (-t)^{s-1} \, e^{-t} \, \mathrm{d} t + \int_{0}^\infty t^{s-1} \, e^{-t} \, \mathrm{d} t  \label{gamma3}
\end{align}
where we can write complementary GF as $\Gamma_c(s)=(-1)^{(s-1)}\Gamma(s)=e^{j\pi(s-1)}\Gamma(s)$. Thus, $|\Gamma_c(s)|=|\Gamma(s)|$. The original GF $\Gamma(s)$ has no zeroes. It has simple poles at zero and at the negative integers (i.e., $s=0,-1,-2,-3, \dots$). However, the proposed GGF $G(s)=\left[(-1)^{s-1}+1\right]\Gamma(s)$ has zeros that occur for even numbers, i.e., for $s=2m, \, m\in \mathbb{Z}$, $G(0)=0$; and $G(s)=2\Gamma(s)$ for $s=2m+1, \, m\in \mathbb{Z}$. It has poles at negative odd integers. For all other values of $s$, it is a complex-valued function. Thus, the GGF is defined over the entire complex plane except at the negative odd integers.

\subsection{Application: Moment generating functions}
Using the traditional LT, the moment generating function (MGFs) are defined only for the pdfs of positive random variables and for finite duration pdfs (such as the uniform pdf). Since all random variables do not have finite moments, MGFs do not always exist. For example, consider a Cauchy random variable, whose moments are not defined because the  integral
\begin{equation}
	\begin{aligned}
		\mathcal{M}_m&= \frac{1}{\pi}\int_{-\infty}^{\infty} \frac{y^m}{1+y^2} \, \mathrm{d} y, \quad m=1,2,\dots,M
	\end{aligned}\label{cauchypdf}
\end{equation} 
do not converge. 
In this Subsection, using the theory of GFT, we show that MGF always exists, even for random variables with non-finite moments.

The moments of a random variable $Y$ are defined as
\begin{align}
E\{Y^m\}=\mathcal{M}_m= \int_{-\infty}^{\infty} y^m f(y) \, \mathrm{d} y, \quad m=1,2,\dots,M,
\end{align}
that can be written as \begin{align}
\mathcal{M}_m= \mathcal{M}_{nm} + \mathcal{M}_{pm}= \int_{0}^{\infty} (-y)^m f(-y) \, \mathrm{d} y+\int_{0}^{\infty} y^m f(y) \, \mathrm{d} y, 
\end{align}
where $\mathcal{M}_{nm}= \int_{0}^{\infty} (-y)^m f(-y) \, \mathrm{d} y$, $ \mathcal{M}_{pm}=\int_{0}^{\infty} y^m f(y) \, \mathrm{d} y$, $f(y)$ is the probability density function (pdf), and $E$ is the expectation operator.
There are many heavy-tailed pdfs whose moments are not finite. To deal with this issue, we define a new moment generating function
as
\begin{align}
M_Y(\sigma,\omega)= \int_{-\infty}^{\infty} f(y) \exp(-\sigma |y|)  \exp(-j\omega y) \, \mathrm{d} y,
\end{align}
that can be further simplified as 
\begin{align}
M_Y(s^*,s)= \int_{0}^{\infty} f(-y)  \exp(-s^* y) \, \mathrm{d} y + \int_{0}^{\infty} f(y)  \exp(-s y) \, \mathrm{d} y. \label{mgf1}
\end{align}
From \eqref{mgf1}, we observe the following moment generation property
\begin{align}
M_Y(s^*,s)= \sum_{m=0}^\infty \frac{(s^*)^m}{m!} \mathcal{M}_{nm}   +  \sum_{m=0}^\infty (-1)^m \frac{s^m}{m!} \mathcal{M}_{pm}   . \label{mgf2}
\end{align}
Thus, 
the proposed MGFs \eqref{mgf2} always exist. Similarly, for a Cauchy random variable, we can define other kind of moments as    
\begin{equation}
	\begin{aligned}
		\mathcal{M}_m&= \frac{1}{\pi}\int_{-\infty}^{\infty} \frac{y^m \exp(-\sigma |y|)}{1+y^2} \, \mathrm{d} y, \quad \sigma>0,
	\end{aligned}\label{cauchypdf1}
\end{equation} 
 which can be written as
\begin{equation}
	\begin{aligned}
		\mathcal{M}_m&=  \left[(-1)^m +1\right] \frac{1}{\pi}\int_{0}^{\infty} \frac{y^m \exp(-\sigma y)}{1+y^2} \, \mathrm{d} y= \frac{2}{\pi}\int_{0}^{\infty} \frac{y^m \exp(-\sigma y)}{1+y^2} \, \mathrm{d} y, \quad \text{ for } m=2,4,6,\dots,
	\end{aligned}\label{cauchypdf2}
\end{equation} 
 and it is zero for the odd values of $m$. The integral \eqref{cauchypdf2} converges for $m>1$ and its solution can be written in terms of the generalized hypergeometric function.

 \section{Generalized discrete time Fourier transform (GDTFT)}
 The discrete time Fourier transform (DTFT) and inverse DTFT pairs of a signal $x[n]$ are defined as
\begin{equation}
	\begin{aligned}
		X(\Omega)&=\sum_{n=-\infty}^{\infty} x[n] \exp(-j\Omega n), \\
		x[n]&={\frac{1}{2\pi}}\int_{-\pi}^{\pi}  X(\Omega) \exp(j\Omega n) \ud \Omega.
	\end{aligned}\label{DFT1}
\end{equation}
 
 Corresponding to GFT, we hereby define the generalized discrete time Fourier transform (GDTFT) as 
 \begin{equation}
 	\begin{aligned}
 		X(\Omega,\sigma)&=\sum_{n=-\infty}^{\infty} x[n] \exp(-\sigma \, \sgn(n)\, n) \exp(-j\Omega n),
 		\end{aligned}\label{DFT2}
 \end{equation}
 which can be written as 
 \begin{align}
 X(\Omega,\sigma)&=\sum_{n=-\infty}^{-1} x[n] \exp((\sigma-j\Omega) n)+\sum_{n=0}^{\infty} x[n] \exp(-(\sigma+j\Omega) n),\\
 	X(z,z^*)&=\mathfrak{X}(z^*)+X(z)=\sum_{n=1}^{\infty} x[-n] (z^*)^{-n}+\sum_{n=0}^{\infty} x[n] z^{-n}, \label{gz1}
 \end{align}
 where $X(z)$ is the unilateral $z$-transform defined as
 \begin{equation}
 	X(z)=\sum_{n=0}^{\infty} x[n] z^{-n}, \label{gz2}
 \end{equation}
 and complementary $z$-transform $\mathfrak{X}(z^*)$ is hereby defined as   
 \begin{equation}
 	\mathfrak{X}(z^*)=\sum_{n=-\infty}^{-1} x[n] (z^*)^{n}=\sum_{n=1}^{\infty} x[-n] (z^*)^{-n}=\sum_{n=0}^{\infty} x[-n] (z^*)^{-n}-x[0], \label{gz3}
 \end{equation}
 where $z=e^{\sigma+j\Omega}=re^{j\Omega}$, its complex conjugate $z^*=e^{\sigma-j\Omega}=re^{-j\Omega}$, $r=e^\sigma$. For $\sigma=0\implies r=1$, GFZT becomes the discrete time Fourier transform (DTFT). 
 The bilateral $z$-transform (BL$z$T) is defined as
\begin{equation}
	X_B(z)=\sum_{n=-\infty}^{\infty} x[n] z^{-n}. \label{bgz2}
\end{equation}

\noindent We can obtain the original signal from the GDTFT by using the inverse GDTFT as 
\begin{align}
 x[n]\exp(-\sigma \, \sgn(n)\,n)&={\frac{1}{2\pi}}\int_{-\pi}^{\pi}  X(\Omega,\sigma) \exp(j\Omega n) \ud \Omega, \label{igdtft}  
\end{align}
and then evaluating \eqref{igdtft} with $\lim \sigma \to 0$.

Next, we present Table II where some signals and their GF$z$T are presented.
{\renewcommand{\arraystretch}{1.5}%
	\begin{table}[h]	 
		\begin{center}
		\caption{Some signals and their GDTFT along with the information whether DTFT and/or BL$z$T exist. We observe that the ROC is the entire $z$-plane, if there are no poles; otherwise, it is always outside of the outermost pole. Also, the DTFT exists, if the ROC includes unit circle  $|z|=1$.} \label{tab2}
			\begin{tabular}{|l|l|l|}
				\hline
				Signals & Proposed GDTFT and ROC & DTFT/BL$z$T \\ \hline
				$\delta[n]$	&  $1$,      all $z$ & Yes/Yes  \\
				$\delta[n-n_0], \, n_0>0$	&  $z^{-n_0}=r^{-n_0}\exp[-j\Omega n_0]$,       $z\ne 0$ & Yes/Yes  \\	
				$\delta[n+n_0], \, n_0>0$	&  $(z^*)^{n_0}=r^{-n_0}\exp[j\Omega n_0]$,       $z \ne 0$ & Yes/Yes  \\	
				1	&  $\frac{(z^*)^{-1}}{1-(z^*)^{-1}}+\frac{1}{1-z^{-1}}$,      $|z|>1$ & Yes/No  \\ 
				$\sgn[n]$	&  $\frac{-(z^*)^{-1}}{1-(z^*)^{-1}}+\frac{z^{-1}}{1-z^{-1}}$,      $|z|>1$ & Yes/No  \\ 
				$\mathrm{u}[n]$	&  $\frac{1}{1-z^{-1}}$,      $|z|>1$ & Yes/Yes  \\ 
				$\mathrm{u}[-n]$	&  $\frac{1}{1-(z^*)^{-1}}$,      $|z|>1$ & Yes/Yes  \\ 		
				$a^{n}$	&  $\frac{(az^*)^{-1}}{1-(az^*)^{-1}}+\frac{1}{1-az^{-1}}$,      $|z|>|a|\, \bigcap \, |z|>1/|a|$ & No/No \\
				
				$a^{|n|}$	&  $\frac{(a^{-1}z^*)^{-1}}{1-(a^{-1}z^*)^{-1}}+\frac{1}{1-az^{-1}}$,      $|z|>|a|$ & No/No \\
				
				$e^{j\Omega_0 n}$	&  $\frac{-(z^*)^{-2}+(z^*)^{-1}\cos(\Omega_0)-j(z^*)^{-1}\sin(\Omega_0)}{1-2(z^*)^{-1}\cos(\Omega_0)+(z^*)^{-2}} + \frac{1-z^{-1}\cos(\Omega_0)+jz^{-1}\sin(\Omega_0)}{1-2z^{-1}\cos(\Omega_0)+z^{-2}}$,      $|z|>1$ & Yes/No \\	
				
				$\cos(\omega_0 n)$	&  $\frac{-(z^*)^{-2}+(z^*)^{-1}\cos(\Omega_0)}{1-2(z^*)^{-1}\cos(\Omega_0)+(z^*)^{-2}} + \frac{1-z^{-1}\cos(\Omega_0)}{1-2z^{-1}\cos(\Omega_0)+z^{-2}}$,      $|z|>1$ & Yes/No  \\	
				$\sin(\omega_0 n)$	&  $\frac{-(z^*)^{-1}\sin(\Omega_0)}{1-2(z^*)^{-1}\cos(\Omega_0)+(z^*)^{-2}} + \frac{z^{-1}\sin(\Omega_0)}{1-2z^{-1}\cos(\Omega_0)+z^{-2}}$,      $|z|>1$ & Yes/No \\
				\hline
			\end{tabular} 
		\end{center}
\end{table}}

\subsection{Properties of the GDTFT}
We present the following properties of the GDTFT: 
\begin{enumerate}
    \item Linearity:
    \begin{align}
        \text{ if } \quad \mathcal{G}\{x[n]\} & =X(z,z^*)=\mathfrak{X}(z^*)+X(z),\\
       \text{ then } \quad \mathcal{G}\left\{\sum_{i} a_i x_i[n]\right\} & =\sum_i a_iX_i(z,z^*)= \sum_i a_i\left[\mathfrak{X}_i(z^*)+X_i(z)\right], \quad  a_i \in \mathbb{C}.
    \end{align}

    \item Time delay: We can write $x[n]=x[n]\,\mathrm{u}[-n-1]+x[n]\,\mathrm{u}[n]$ and hence, 
    \begin{align}
        \mathcal{G}\{x[n+n_0] \, \mathrm{u}[-n-1-n_0]+x[n-n_0] \, \mathrm{u}[n-n_0]\} & ={(z^*)^{-n_0}}\mathfrak{X}(z^*)+z^{-n_0}X(z), \quad n_0\ge0,
    \end{align}
    and
    \begin{align}
        \mathcal{G}\{x[n-m]\} & =\sum_{-\infty}^{-1} x[n-m) \, {(z^*)^n} +\sum_{0}^{\infty}x[n-m]\,z^{-n}, \quad m\ge 0,\nonumber\\
        & ={(z^*) ^{m}}\sum_{-\infty}^{-1-m} x[n] \, {(z^*)^n} +{z^{-m}}\sum_{-m}^{\infty}x[n]\,z^{-n},\\
        &=\left[-\sum_{\ell=-m}^{-1}x[\ell]\,(z*)^{\ell+m}+(z*)^{m}X(z*)\right] + \left[\sum_{\ell=-m}^{-1}x[\ell]\,z^{-\ell-m}+z^{-m}X(z)\right].
    \end{align}
    
    \item Time advance: 
    \begin{align}
        \mathcal{G}\{x[n+m]\} & =\sum_{-\infty}^{-1} x[n+m) \, {(z^*)^n} +\sum_{0}^{\infty}x[n+m]\,z^{-n}, \quad m\ge 0,\nonumber\\
        & ={(z^*) ^{-m}}\sum_{-\infty}^{-1+m} x[n] \, {(z^*)^n} +{z^{m}}\sum_{m}^{\infty}x[n]\,z^{-n},\\
        &=\left[\sum_{\ell=0}^{m}x[\ell]\,(z*)^{\ell-m}+(z*)^{-m}X(z*)\right] + \left[-\sum_{\ell=0}^{m}x[\ell]\,z^{-\ell+m}+z^{m}X(z)\right].
    \end{align}
    
    \item Time reversal:
    \begin{align}
        \mathcal{G}\{x[-n]\} & =X(z^*)+\mathfrak{X}(z).
    \end{align}

    \item Time expansion (scaling):
    \begin{align}
     \mathcal{G}\{x[n/k]\} & =\mathfrak{X}\left((z^*)^k\right)+X\left(z^k\right), \quad k\ge 1.
    \end{align}

    \item The GDTFT of successive time difference:
    \begin{align}
    \mathcal{G}\{x[n]-x[n-1]\} & =\left[(1-z^*)\mathfrak{X}(z^*)+x(-1)\right]+\left[(1-z^{-1})X(z)-x(-1)\right].
\end{align}

\item The GDTFT of accumulation across time:  
\begin{align}
 \mathcal{G}\left\{\sum_{-\infty}^n x[k] \right\} & =\frac{\mathfrak{X}\left(z^*\right)}{1-z^*}+\frac{X\left(z\right)}{1-z^{-1}}.
\end{align}

\item The GDTFT of $x[n]$ multiplied by various functions, i.e., amplitude scaling:
\begin{align}
 \mathcal{G}\{a^{-n}\,x[n]\} & =\mathfrak{X}(a^{-1}z^*)+X(a z),\\
 \mathcal{G}\{a^{-|n|}\,x[n]\} & =\mathfrak{X}(a z^*)+X(a z),\\
  \mathcal{G}\{n^{m}\,x(t)\} & ={(z^*)}^{-m}\frac{\mathrm{d}^m}{\mathrm{d}{z^*}^m }\mathfrak{X}(z^*)+(-z)^m\frac{\mathrm{d}^m}{\mathrm{d}z^m }X(z),\\  
   \mathcal{G}\{|n|^{m}\,x[n]\} & =(-z^*)^{-m}\frac{\mathrm{d}^m}{\mathrm{d}{z^*}^m }\mathfrak{X}(z^*)+(-z)^{m}\frac{\mathrm{d}^m}{\mathrm{d}z^m }X(z).   
\end{align}
\end{enumerate}

\section{Conclusion}\label{con}
Important and fundamental contributions of this study are the introduction of the generalized Fourier transform (GFT) in both continuous-time and discrete-time domains, the Fourier scale transform, and the solution of initial value problem using the GFT and Fourier transform.  
This work has presented the GFT, its properties, discrete-time version, comparison with Laplace and $z$-transforms, along with many interesting aspects including solution of initial value problems. The most interesting aspect of this newly proposed GFT is that it overcomes the limitations generally observed in Fourier, Laplace and $z$ transforms. It is applicable to a much larger class of signals and simplifies the analysis of both continuous-time and discrete-time signals and systems. This transform may find greater utility in applications in the near future. 



\appendix
\section{Multi-dimensional GFT}
We can define the $M$-dimensional ($M$D) GFT as
\begin{align}
    G(\omega_1,\omega_2,\dots,\omega_M,\sigma)&=\int_{-\infty}^{\infty}\int_{-\infty}^{\infty} g(x_1,x_2,\dots,x_M) \exp(-\sigma(|x_1|+|x_2|+\dots+|x_M|)) \times \nonumber \\ &  \exp(-j\omega_1\,x_1-j\omega_2\,x_2-\dots-j\omega_M x_M) \mathrm{d}x_1 \, \mathrm{d}x_2 \dots \mathrm{d}x_M, \label{MD1}
\end{align}
where the value of $\sigma$ is such that the integral \eqref{MD1} exists which resolves the region of convergence (ROC), i.e., $\sigma\in \text{ROC}$ that depends on the function $g(x_1,x_2,\dots,x_M)$.
We can recover the original signal by evaluating
\begin{align}
   g(x_1,x_2,\dots,x_M)  \exp(-\sigma(|x_1|+|x_2|+\dots+|x_M|)) &=\frac{1}{(2\pi)^M}\int_{-\infty}^{\infty}\int_{-\infty}^{\infty}  G(\omega_1,\omega_2,\dots,\omega_M,\sigma) \times \nonumber \\ &  \exp(j\omega_1\,x_1+j\omega_2\,x_2-\dots+j\omega_M x_M) \mathrm{d}\omega_1 \, \mathrm{d}\omega_2 \dots \mathrm{d}\omega_M, \label{MD2}
\end{align}
and then considering, $\lim \sigma \to 0$. Using, $s_k=\sigma+j\omega_k$, we can write \eqref{MD1} and \eqref{MD2} in terms of symbol $s_k$ where $k=1,2,\dots,M$. Clearly, the $M$D-GFT becomes the $M$D-FT when $\lim \sigma \to 0$, and if $g(x_1,x_2,\dots,x_M)$ is defined only for $x_1,x_2,\dots,x_M \ge 0$, and zero otherwise, then it is the $M$D-LT \cite{MDLT1}.



\end{document}